\documentclass[journal]{IEEEtran}

\usepackage{amsmath}
\usepackage{amssymb}
\usepackage{amsfonts}
\usepackage{amscd}
\usepackage{amsthm}
\usepackage[noadjust]{cite}

\usepackage{mathrsfs}
\usepackage{bbm}
\usepackage{amsbsy}

\usepackage{scrextend}

\usepackage{xfrac}

\newtheorem{theorem}{Theorem}
\newtheorem{lemma}{Lemma}

\theoremstyle{definition}
\newtheorem{example}{Example}

\makeatletter
\renewcommand*\env@matrix[1][c]{\hskip -\arraycolsep
  \let\@ifnextchar\new@ifnextchar
  \array{*\c@MaxMatrixCols #1}}
\makeatother
\usepackage{graphicx}

\usepackage{etoolbox}
\AtEndEnvironment{example}{\null\hfill\IEEEQED}%
\AtEndEnvironment{definition}{\null\hfill\IEEEQED}%

\newcommand{\tth}{{\text{th}}} 
\newcommand{\tr}{{\sf T}} 

\newcommand{\Eb}{\mathbb{E}}

\newcommand{\Zb}{\mathbb{Z}}
\newcommand{\Cb}{\mathbb{C}}
\newcommand{\p}{\mathbf}

\newcommand{\na}{{n_a}}
\newcommand{\nt}{{n_t}}
\newcommand{\nr}{{n_r}}
\newcommand{\snr}{{\sf SNR}}
\newcommand{\Cs}{\mathcal{C}}
\newcommand{\Cn}{\mathcal{C}_{{\sf T}}}
\newcommand{\Mb}{\mathbbmss{M}}
\newcommand{\Mm}{\mathbbmss{\widetilde{M}}}
\newcommand{\Ac}{\mathcal{A}}
\newcommand{\Tc}{\mathcal{T}}
\newcommand{\spe}{\eta}
\newcommand{\gsm}{{\sf GSM}}
\newcommand{\Pn}{P(\Cn)} 
\newcommand{\energy}{\mathcal{E}}
\newcommand{\dmin}{d_{\min}}

\title{\LARGE Improving Generalized Spatial Modulation using Translation Patterns}
\author{Lakshit Singla and Lakshmi Natarajan
\thanks{The authors are with the Department of Electrical Engineering, Indian Institute of Technology Hyderabad, email: \{ee17btech11021,\,lakshminatarajan\}@iith.ac.in.}%
\thanks{
The authors thank Dr.\ Lakshmi Narasimhan Theagarajan, Indian Institute of Technology Palakkad, for several insightful discussions.
}%
}

\begin{document}

\maketitle

\begin{abstract}
\boldmath
Generalized spatial modulation (GSM) is a spectral-efficient technique used in multiple-input multiple-output (MIMO) wireless communications when the number of radio frequency chains at the transmitter is less than the number of transmit antenna elements.
We propose a family of signal constellations, as an improvement over GSM, which splits the information bits into three parts, and encodes the first part into a set of complex symbols, the second part into the choice of a subset of antennas activated for transmission (as in GSM), and the third into a \emph{translation pattern} that offsets the symbols transmitted through the activated antennas.
The nominal coding gain (the ratio of the squared minimum distance between transmit vectors to the transmit power) of our scheme is higher than that of GSM by at least $0.86$~dB, and this improvement can be as much as $2.87$~dB based on the system parameters.
We show that the new scheme has advantages over other known signal constellations for GSM, in terms of error performance, nominal coding gain and design flexibility.
\end{abstract}

\begin{IEEEkeywords}
coding gain, generalized spatial modulation (GSM), MIMO, minimum distance 
\end{IEEEkeywords}

\section{Introduction} \label{sec:introduction}

\IEEEPARstart{T}{he} high complexity involved in the design of communication terminals with multiple radio frequency (RF) chains can lead to scenarios where the number of RF chains $\na$ available in a wireless transmitter is less than the number of antenna elements $\nt$~\cite{RHG_COMM_MAG_11}.
Generalized spatial modulation (GSM)~\cite{WJS_TWCOM_12,YSMH_Asilomar_10,FHXYH_Globecom_10} is a well-known modulation technique that provides high spectral efficiencies in such scenarios.
In GSM, the choice of $\na$ antenna elements activated for transmission out of the $\nt$ available elements, called \emph{antenna activation pattern}, conveys a part of the information bits.


Several signal constellations have been proposed in the literature that improve the error performance of GSM~\cite{CSSS_TCOM_15,CSSS_TWC_16,FRS_WCL_18,FrS_COMML_17,GZZDLA_TWC_19}, albeit with some design limitations. 
Enhanced Spatial Modulation (ESM) and its variants~\cite{CSSS_TCOM_15,CSSS_TWC_16} have been designed for limited choice of system parameters (for example, complex symbols must take values from $16$- or $64$-QAM and $\na$ must be even for ESM-Type2 constellations~\cite{CSSS_TWC_16}).
The signal constellation in~\cite{FRS_WCL_18} relies on computer search and is available only for \mbox{$\na=2$}.
The scheme from~\cite{FrS_COMML_17} does not modulate an integer number of information bits, and its transmit vectors contain irrational components.
Note that, to minimize the cost of digital-to-analog conversion at the transmitter it is desirable to use signals that do not contain irrational components~\cite{HaV_WCL_13}.
The signal design technique of~\cite{GZZDLA_TWC_19} relies on a numerical optimization algorithm whose complexity increases exponentially in the spectral efficiency.

In this paper, we exploit a technique which is the essence of \emph{Construction A} of lattices~\cite{CoS_Springer_13} to obtain new signal constellations for GSM.
This technique uses the codewords of an error correcting code~\cite{MWS_Theory_77} to generate high-density lattice packings.
In our proposed communication scheme we partition information bits into three blocks, the first is modulated into complex symbols, the second into an antenna activation pattern, and the third into a \emph{translation pattern} or \emph{vector} that offsets the complex symbols transmitted through the activated antennas.
The translation patterns used in our scheme arise from the codewords of the single-parity check code.
We use the ratio of the squared minimum distance of the constellation to the transmit power, which we refer to as the \emph{nominal coding gain}, as the figure-of-merit of a scheme. 
We show that the proposed scheme provides improvements in nominal coding gain over GSM by at least $0.86$~dB and up to $2.87$~dB.

In comparison with~\cite{CSSS_TCOM_15,CSSS_TWC_16,FRS_WCL_18,FrS_COMML_17,GZZDLA_TWC_19}, our technique allows a simple and explicit construction for any choice of alphabet size (for the complex symbols), and any choice of $\nt, \na \geq 2$.
This flexibility allows us to achieve a wide range of spectral efficiencies, including the ability to encode integer number of information bits.
Further, the transmit vectors in our constellation do not contain irrational components. 
Simulation and numerical results show that our scheme provides a lower error rate and larger nominal coding gain than GSM and the schemes from~\cite{CSSS_TCOM_15} and~\cite{FRS_WCL_18}, and the error performance is similar to that of the scheme proposed in~\cite{CSSS_TWC_16}.


\noindent
\emph{Notation:}
We denote column vectors and matrices using bold lower and upper case letters, respectively. 
For any vector $\p{x}$, the symbols $\p{x}^\tr$, $\|\p{x}\|$ and $\|\p{x}\|_0$ denote the transpose, Euclidean norm ($\ell_2$-norm) and the $\ell_0$-norm of $\p{x}$, respectively.
For a positive integer $N$, $[N]$ denotes the set $\{1,\dots,N\}$. 
For sets $A,B$, the symbol $A \setminus B$ is the set of all elements in $A$ that are not in $B$.
$\mathcal{CN}(0,\sigma^2)$ is the circularly symmetric complex Gaussian distribution with zero mean and variance $\sigma^2$.

\section{System Model and Preliminaries} \label{sec:system_model}

We consider a wireless slow fading channel with $\nt$ transmit and $\nr$ receive antennas, modelled as \mbox{$\p{y} = \p{Hx} + \p{w}$},
where \mbox{$\p{x} \in \Cb^{\nt}$} is the transmit vector, \mbox{$\p{H} \in \Cb^{\nr \times \nt}$} is the channel matrix, \mbox{$\p{y} \in \Cb^{\nr}$} is the received vector and \mbox{$\p{w} \in \Cb^{\nr}$} is the additive noise at the receiver.
We assume Rayleigh fading, i.e., each entry of $\p{H}$ is an independent $\mathcal{CN}(0,1)$ random variable, and that only the receiver knows the value of $\p{H}$ and the transmitter does not.
The entries of $\p{w}$ are independent $\mathcal{CN}(0,N_o)$ random variables.
In this paper we consider modulation schemes for this channel that span only one time slot, i.e., we do not consider coding across time.
We further assume that the transmitter is equipped with $\na$ RF chains, where $2 \leq \na \leq \nt$. 
This implies that the number of active antennas, at any given time instance is at the most $\na$, i.e., any transmit vector $\p{x}$ must satisfy $\|\p{x}\|_0 \leq \na$.

A \emph{signal constellation} $\Cs$ for $\nt$ transmit antennas and $\na$ RF chains is a finite set of complex vectors each of length $\nt$ and with $\ell_0$-norm less than or equal to $\na$.
We assume that each vector in $\Cs$ is equally likely to be transmitted.
The spectral efficiency of $\Cs$ is \mbox{$\eta=\log_2 |\Cs|$}~bits/sec/Hz, the average transmit power is $P = \sum_{\p{x} \in \Cs} \|\p{x}\|^2/|\Cs|$, and the resulting signal to noise ratio is $\snr = P/N_o$.
The minimum Euclidean distance between the points in $\Cs$ is $\dmin(\Cs) = \min\{\,\|\p{x} - \p{x'}\|~|~\p{x}, \p{x'} \in \Cs, \, \p{x} \neq \p{x'}\,\}$.
We will assume that the receiver employs the maximum likelihood detector, i.e., the detector output is 
$\arg \min_{\p{x} \in \Cs} \|\p{y} - \p{Hx}\|^2$.
The pair-wise error probability between $\p{x}, \p{x'} \in \Cs$ 
is upper bounded by~\cite{TsV_Cambridge}
\begin{equation*} 
\frac{(4 N_o)^\nr}{\|\p{x} - \p{x'}\|^{2\nr}} = \frac{4^\nr \, P^{\nr}}{\snr^\nr \! \|\p{x} - \p{x'}\|^{2\nr}} \leq  \frac{4^\nr}{\snr^\nr} \cdot \left(\frac{P}{\dmin^2(\Cs)}\right)^\nr \!\!\!\!\!\!.
\end{equation*} 
The \emph{nominal coding gain} $\delta(\Cs)$ of the constellation $\Cs$ is 
\mbox{$\textstyle \delta(\Cs) = {\dmin^2(\Cs)}/{P}$}.
We deduce that the pairwise error probability between any pair of transmit vectors is at the most 
$\textstyle \left(\,{4}\,/ \,{\snr \, \delta(\Cs)} \,\right)^\nr$. 
Typically, a large value of $\delta(\Cs)$ indicates a small probability of error. 

\subsection{Review of Generalized spatial modulation~\cite{WJS_TWCOM_12,YSMH_Asilomar_10,FHXYH_Globecom_10}}

In GSM, $\na$ information-bearing complex symbols are transmitted through a set of $\na$ antennas among the available $\nt$ transmit antennas. 
The choice of the antennas transmitting these symbols carries additional information. 
Let \mbox{$\Ac=\{A_1,\dots,A_L\}$} be a collection of $L$ distinct $\na$-sized subsets of $[\nt]$, i.e., $A_\ell \subset [\nt]$ and $|A_\ell| = \na$ for each \mbox{$\ell \in [L]$}.
Each $A \in \Ac$ denotes a possible collection of $\na$ antennas 
that can be activated for transmission.
Clearly, $L \leq \binom{\nt}{\na}$.
It is not uncommon to choose $L$ as the largest power of $2$ less than or equal to $\binom{\nt}{\na}$, i.e., $\log_2 L = \lfloor \log_2 \binom{\nt}{\na}\rfloor$.
If an $M$-ary complex alphabet \mbox{$\Mb \subset \Cb$}, such as $M$-QAM, is used for encoding, then $\na \log_2M$ information bits are modulated into a collection of $\na$ symbols $z_1,\dots,z_\na \in \Mb$. 
An additional $\log_2 L$ information bits are encoded into the subset $A \in \Ac$ of antennas activated for transmission.
Then the symbols $z_1,\dots,z_{\na}$ are transmitted through the subset $A$ of antennas.
Hence, the spectral efficiency of GSM is
$\spe_\gsm = \na \log_2 M + \log_2 L \text{ bits/sec/Hz}$.

Throughout this paper, we will refer to the set \mbox{$\Cs \subset \Cb^{\nt}$} of transmit vectors as \emph{signal constellation} or simply the \emph{constellation}. 
We will also use complex signal sets, such as \mbox{$\Mb \subset \Cb$}, to design new constellations. We will refer to such complex signal sets as \emph{complex alphabets} or simply \emph{alphabets}.

\section{Proposed Signal Constellation} \label{sec:C3}

We present the new family of signal constellations in this section for any choice of $\nt$ and $2 \leq \na \leq \nt$. 

\subsection{Complex Alphabets}

We denote the set of Gaussian integers as $\Zb[i]$, i.e., $\Zb[i]=\{a + ib~|~ a,b \in \Zb\}$, and denote the complex number $\sfrac{1}{2} + \sfrac{i}{2}$ as $\alpha$. The set \mbox{$\Zb[i] + \alpha = \{x + \alpha~|~x \in \Zb[i]\,\}$}
consists of all complex numbers whose real and imaginary parts are \emph{half-integers}, i.e., are of the form $q/2$ where $q$ is an odd integer.
Note that $\Zb[i] + \alpha$ does not contain $0$, and the smallest squared absolute value among the elements in $\Zb[i] + \alpha$ is $\sfrac{1}{2}$, which is achieved by $\pm \sfrac{1}{2} \pm \sfrac{i}{2}$. 
Our construction makes use of complex alphabets $\Mm$ that satisfy the following two properties
\begin{equation} \label{eq:P1_P2}
\text{\it Property P1. } \Mm \subset \Zb[i] + \alpha, \text{ and } \text{\it Property P2. } -\alpha \notin \Mm.
\end{equation}


\begin{figure}
\begin{flushleft}
\setlength{\tabcolsep}{3pt}
\begin{tabular}{ll}
\includegraphics[width=1.65in]{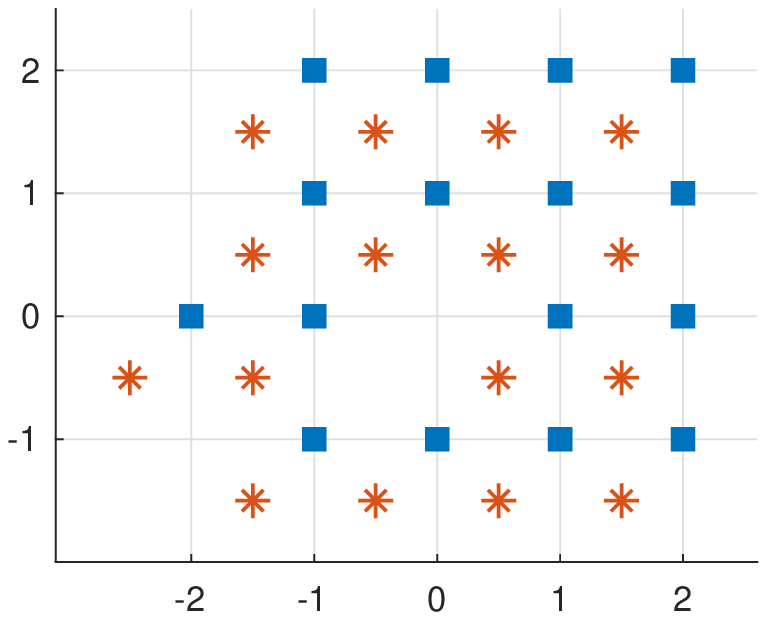}
&
\includegraphics[width=1.75in]{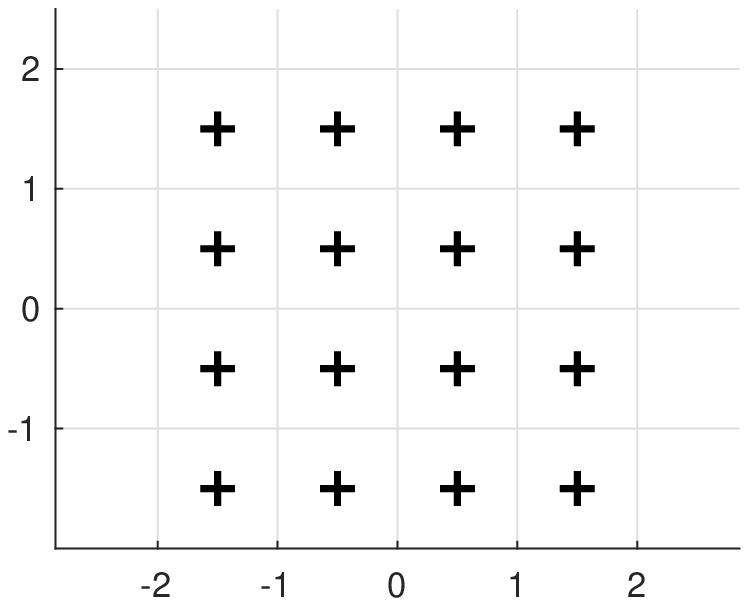}
\end{tabular}
\end{flushleft}
\vspace{-2mm}
\caption{Left: modified square $16$-QAM alphabet $\Mm$ (red asterisks), and its translation $\Mm + \alpha$ (blue squares) that will be used in our constellation design. Right: conventional $16$-QAM. Axes are real and imaginary parts.}
\label{fig:modified_16QAM}
\vspace{-4mm}
\end{figure}

{
\begin{example}
\emph{Modified square QAM alphabets:} Let $M$ be such that $\sqrt{M}$ is an even integer. The conventional $M$-ary square QAM alphabet $\Mb$ is the set of all complex numbers whose real and imaginary parts belong to $\sqrt{M}$-ary PAM, i.e.,
$\textstyle \{ \frac{-\sqrt{M}+1}{2},\frac{-\sqrt{M}+3}{2},\dots,\frac{-1}{2},\frac{1}{2},\dots,\frac{\sqrt{M}-1}{2} \}$. 
Note that \mbox{$\Mb \subset \Zb[i] + \alpha$} and \mbox{$-\alpha = -\sfrac{1}{2} -\sfrac{i}{2} \in \Mb$}.
Thus, $\Mb$ satisfies property \emph{P1}, but not \emph{P2}.
We design $\Mm$ by starting with $\Mb$, and replacing the element $-\alpha$ with an element \mbox{$\beta \in (\Zb[i] + \alpha) \setminus \Mb$}, i.e.,
\mbox{$\Mm = \Mb \cup \{\beta\} \setminus \{-\alpha\}$}.
In order to reduce the transmit power, we choose $\beta=-\sfrac{(\sqrt{M}+1)}{2} - \sfrac{i}{2}$. 
We refer to this resulting alphabet $\Mm$ as \emph{modified square QAM}.
It is clear that $\Mm$ satisfies properties \emph{P1} and \emph{P2} in~\eqref{eq:P1_P2}.
The modified $4$-QAM alphabet consists of the points $\sfrac{1}{2} + \sfrac{i}{2}$, $\sfrac{1}{2} - \sfrac{i}{2}$, $-\sfrac{1}{2} + \sfrac{i}{2}$ and $-\sfrac{3}{2} - \sfrac{i}{2}$, where the last element is $\beta$. The modified $16$-QAM alphabet is shown in Fig.~\ref{fig:modified_16QAM}. 
\end{example}

For an arbitrary value of $M$, we can obtain an alphabet $\Mm$ satisfying~\eqref{eq:P1_P2} using an inexpensive computer search. To reduce the transmit power we require points $x \in (\Zb[i] + \alpha) \setminus \{-\alpha\}$ for which $|x|^2 + |x+\alpha|^2$ is small; see~\eqref{eq:Pn}.
We obtain the alphabet $\Mm$ by sorting the elements $x$ of the set $\{a+ib~|~|a|, |b| \in \{\frac{1}{2}, \frac{3}{2},\dots, \lceil \frac{M}{2}\rceil  + \frac{1}{2} \}\} \setminus \{-\alpha\}$ in the ascending order of $|x|^2 + |x+\alpha|^2$, and choosing the first $M$ elements in this list.
}

In order to analyze the transmit power of our proposed constellation, we require the values of the average energies of the symbols in the alphabets $\Mm$ and $\Mm+\alpha=\{x + \alpha~|~x \in \Mm\}$.
When $\Mm$ is the modified square $M$-QAM, straightforward calculations show that the energies of these alphabets are
\begin{align}
\textstyle 
\energy(\Mm) &= \frac{1}{M} \sum_{x \in \Mm} |x|^2 =\frac{M-1}{6} + \frac{1}{4} + \frac{1}{2\sqrt{M}}, \label{eq:energy_Mm} \\
\energy(\Mm + \alpha) &= \frac{1}{M} \!\! \sum_{x \in \Mm + \alpha} \!\!\!\! |x|^2 = \frac{M-1}{6} + \frac{3}{4}. \label{eq:energy_Mm_alpha} 
\end{align} 
In comparison, the energy of the conventional square $M$-QAM is $(M-1)/6$.

\subsection{New Constellation using Translation Patterns}

{
Our new constellation uses the following ingredients:
\begin{enumerate}
\item[\emph{(i)}] Any $M$-ary complex alphabet $\Mm$ satisfying properties \emph{P1} and \emph{P2} in~\eqref{eq:P1_P2}. 
\item[\emph{(ii)}] Any collection $\Ac=\{A_1,\dots,A_L\}$ of antenna activation patterns, where each $A_\ell$ is an $\na$-sized subset of $[\nt]$.
\item[\emph{(iii)}] A set $\Tc$ of \emph{translation vectors} or \emph{translation patterns}, described below, that consists of $2^{\na-1}$ complex vectors, each of length $\na$, i.e., $\Tc \subset \Cb^\na$ and $|\Tc|=2^{\na-1}$. 
\end{enumerate} 

We restrict all the components of every translation vector to be equal to either $0$ or $\alpha$, and let $\Tc$ be the set of all such vectors whose components contain an even number of $\alpha$'s, i.e.,
\mbox{$\Tc = \left\{\, \p{t} \in \{0,\alpha\}^\na~|~\|\p{t}\|_0 \text{ is even}\, \right\}$}.
Note that $\Tc$ has the same structure as the binary single-parity check code~\cite{MWS_Theory_77}, except that the $1$'s in the single-parity check code are replaced with $\alpha$'s.
From the well-known properties of the single-parity check code, we deduce that \mbox{$|\Tc|=2^{\na-1}$}, and 
for any \mbox{$l \in [\na]$} if a translation vector is picked uniformly at random from $\Tc$, the value of its $l^\tth$ component is equally likely to be $0$ or $\alpha$.

Our scheme encodes $\na \log_2 M + \log_2 L + \na - 1$ information bits. 
The first $\na \log_2 M$ bits are modulated, $\log_2 M$ bits at a time, into the complex symbols $z_1,\dots,z_{\na} \in \Mm$, yielding the vector $\p{z}=(z_1,\dots,z_\na)^\tr$. The next block of $\log_2 L$ information bits select the antenna activation pattern $A \in \Ac$. The remaining \mbox{$\na-1$} information bits, say $b_1,\dots,b_{\na-1}$, are encoded into a single-parity check codeword, and the $1$'s appearing in this codeword are converted into $\alpha$'s to obtain the translation pattern $\p{t} \in \Tc$, i.e., $\p{t} = \alpha \times (b_1,\dots,b_{\na-1}, \oplus_{i=1}^{\na-1} b_i)^\tr$, where $\oplus$ denotes the binary XOR operation. 

We translate the vector $\p{z}$ using $\p{t}$, and transmit the $\na$ components of the resulting vector $\p{z} + \p{t}$ through the antennas with indices in $A$.
%
Let $A = \{j_1,\dots,j_{\na}\}$ with $j_1 < j_2 < \cdots < j_{\na}$.
The function $\p{x} = \psi_A(\p{z} + \p{t})$, defined below in~\eqref{eq:psi_A}, maps the entries of $\p{z} + \p{t}$ into the components of $\p{x}$ indexed by $A$,
\begin{equation} \label{eq:psi_A}
x_{j_l} = z_l + t_l \text{ for all } l \in [\na],~~ \text{ and }~~ x_j = 0 \text{ if } j \notin A.
\end{equation} 
The proposed constellation $\Cn(\Mm,\Ac,\Tc)$, or simply $\Cn$, is
\begin{equation} \label{eq:Cn}
\textstyle \Cn = \bigcup_{A \in \Ac} \, \bigcup_{\,\p{t} \in \Tc} \, \left\{ \psi_A(\p{z} + \p{t})~\Big|~ z_l \in \Mm, \, l \in [\na]  \,\right\}.
\end{equation}

The process of modulating $\na \log_2 M$ bits into complex symbols $z_1,\dots,z_{\na}$, and using $\log_2 L$ bits to select an antenna activation pattern $A$ are similar to GSM.
The additional steps used by our scheme are computing the parity $\oplus_{i=1}^{\na-1} b_i$ and the sum $\p{z} + \p{t}$, which have negligible complexity.
}

\begin{example} \label{ex:Example2}
Consider \mbox{$\nt=4$} and \mbox{$\na=3$}. 
In this case, $\Tc$ consists of the four vectors $(0,0,0)^\tr$, $(0,\alpha,\alpha)^\tr$, $(\alpha,0,\alpha)^\tr$, $(\alpha,\alpha,0)^\tr$.
Let \mbox{$L=4$} and $\Ac$ be the collection of the four subsets $\{1,2,3\}, \{1,2,4\}, \{1,3,4\}, \{2,3,4\}$.
The number of information bits encoded by $\Cn$ is $\na \log_2 M + \log_2 L + \na - 1= 3\log_2 M + 4$. Out of these, $3\log_2 M$ bits are encoded into complex symbols $z_1,z_2,z_3 \in \Mm$, $2$ bits are used to select one of the activation patterns from $\Ac$, and the remaining $2$ bits are used to choose one of the four possible translation vectors in $\Tc$. If the chosen translation vector is $(t_1,t_2,t_3)^\tr$, then the transmitted vector $\p{x}$ is of the form
\begin{align*}
&(z_1 + t_1, z_2 + t_2, z_3 + t_3, 0)^\tr, ~(z_1 + t_1, z_2 + t_2, 0 , z_3 + t_3)^\tr, \\
&(z_1 + t_1, 0, z_2 + t_2, z_3 + t_3)^\tr \text{ or } (0, z_1 + t_1, z_2 + t_2, z_3 + t_3)^\tr
\end{align*}
which correspond to the antenna activation patterns $\{1,2,3\}$, $\{1,2,4\}$, $\{1,3,4\}$ and $\{2,3,4\}$, respectively. 
\end{example}



\subsubsection*{Transmit Power of $\Cn$}
Let $\Eb$ denote the expectation operation.
We assume that the $\na$ components of $\p{z}$ are uniformly distributed over $\Mm$, $\p{t}$ has uniform distribution over $\Tc$, and all the activation patterns $A$ in $\Ac$ are equally likely. 
Since $\p{x}=\psi_A(\p{z} + \p{t})$, we observe that $\|\p{x}\|^2 = \|\p{z} + \p{t}\|^2$. Thus, the power of the constellation $\Cn$ is
$\Pn = \Eb \, \|\p{x}\|^2 = \Eb \, \|\p{z} + \p{t}\|^2 = \sum_{l=1}^{\na}\, \Eb \,|z_l + t_l|^2$.
Note that for each $l \in [\na]$, $z_l$ and $t_l$ are independent, $z_l$ is uniformly distributed in $\Mm$, and $t_l$ is equally likely to be $0$ or $\alpha$. Hence, $\Eb \,|z_l|^2=\energy(\Mm)$ and $\Eb \, |z_l + \alpha|^2=\energy(\Mm+\alpha)$, and 
\begin{equation} \label{eq:Pn}
\textstyle
\Pn = \sum_{l=1}^{\na} \frac{\Eb\,|z_l|^2 + \Eb\,|z_l + \alpha|^2}{2} = \na \, \frac{\energy(\Mm) + \energy(\Mm+\alpha)}{2}.
\end{equation} 

\subsubsection*{Minimum Squared Distance of $\Cn$}
We will show that $\dmin^2(\Cn)=1$ through the following lemmas. We will assume that the alphabet $\Mm$ satisfies the properties in~\eqref{eq:P1_P2}. 

\begin{lemma} \label{lem:x_supp}
For any $\p{z} \in \Mm^\na$, $\p{t} \in \Tc$ and $A \in \Ac$, the vector $\p{x} = \psi_A(\p{z} + \p{t})$ satisfies $|x_j|^2 \geq \sfrac{1}{2}$ for all $j \in A$ and $x_j=0$ for all $j \notin A$.
\end{lemma}
\begin{proof}
For a given \mbox{$j \in A$} there exists an $l \in [\na]$ such that \mbox{$x_j = z_l + t_l$}. 
If \mbox{$t_l=0$}, then \mbox{$x_j=z_l \in \Mm \subset \Zb[i] + \alpha$}, and since every element in $\Zb[i]+\alpha$ has squared absolute value at least $\sfrac{1}{2}$, we have $|x_j|^2 = |z_l|^2 \geq \sfrac{1}{2}$.
On the other hand, if $t_l=\alpha$, then the real and imaginary parts of $x_j = z_l + \alpha$ are both integers, and hence, $x_j \in \Zb[i]$. Also, $z_l \neq -\alpha$ since $-\alpha \notin \Mm$, and hence, $x_j \neq 0$. In this case, $|x_j|^2 \geq 1$. We conclude that $|x_j|^2 \geq \sfrac{1}{2}$ for all $j \in A$. 
Also, by construction, $x_j=0$ for $j \notin A$, see~\eqref{eq:psi_A}. 
\end{proof}

\begin{lemma} \label{lem:z_t_distance}
Let $\p{z},\p{z'} \in \Mm^\na$ and $\p{t},\p{t'} \in \Tc$ be any choice of vectors such that $(\p{z},\p{t}) \neq (\p{z'},\p{t'})$. 
Then $\|\p{z}+\p{t} - (\p{z'} + \p{t'})\|^2 \geq 1$.
\end{lemma}
\begin{proof}
Since $\Mm \subset \Zb[i] + \alpha$, the real and imaginary parts of every component of $\p{z}$ and $\p{z'}$ are half-integers, i.e., belong to $\Zb + \sfrac{1}{2}$. 
This implies that the real and imaginary parts of each component of $\p{z}-\p{z'}$ is an integer, i.e., $\p{z} - \p{z'} \in \Zb[i]^{\na}$.
 
We first consider the case $\p{t} = \p{t'}$. Necessarily $\p{z} \neq \p{z'}$, and at least one of the components of $\p{z} - \p{z'}$ is non-zero. Since this component of $\p{z} - \p{z'}$ belongs to $\Zb[i]$ and is non-zero, we have $\|\p{z} + \p{t} - (\p{z'} + \p{t'})\|^2 = \|\p{z} - \p{z'}\|^2 \geq 1$.

Now consider $\p{t} \neq \p{t'}$. Since the vectors $\p{t}$ and $\p{t'}$ have even number of non-zero elements (i.e., $\alpha$'s), they differ in at least two components, say $l,m \in [\na]$. The $l^\tth$ and $m^\tth$ components of $\p{t} - \p{t'}$ belong to the set $\{\alpha,-\alpha\}$.
Now consider the vector $\p{z} + \p{t} - (\p{z'} + \p{t'}) = (\p{z} - \p{z'}) + (\p{t} - \p{t'})$. The $l^\tth$ component of $\p{z}-\p{z'}$ belongs to $\Zb[i]$ and the $l^\tth$ component of $\p{t}-\p{t'}$ is $\pm \alpha$. Thus, the $l^\tth$ component of $(\p{z} - \p{z'}) + (\p{t} - \p{t'})$ belongs to $\Zb[i] + \alpha$, and therefore, this component has squared magnitude at least $\sfrac{1}{2}$. Similar result holds for the $m^\tth$ component of $(\p{z} - \p{z'}) + (\p{t} - \p{t'})$ as well. 
We conclude that the squared norm of this vector is at least $\sfrac{1}{2} + \sfrac{1}{2} = 1$.
\end{proof}


\begin{theorem} \label{thm:mindist_Cn}
The minimum distance of $\Cn$ is $\dmin(\Cn)=1$.
\end{theorem}
\begin{proof}
We will show that $\dmin^2(\Cn) \geq 1$ and $\dmin^2(\Cn) \leq 1$.

\emph{Lower bound:} Let $\p{x}$ and $\p{x'}$ be generated using the tuples $(\p{z},\p{t},A)$ and $(\p{z'},\p{t'},A')$, respectively, i.e., $\p{x}=\psi_A(\p{z} + \p{t})$ and $\p{x'} = \psi_{A'}(\p{z'} + \p{t'})$. We will assume that $(\p{z},\p{t},A) \neq (\p{z'},\p{t'},A')$.

Assume $A \neq A'$. 
There exist $j \in A \setminus A'$ and $j' \in A' \setminus A$, since $A \neq A'$ and $|A|=|A'|$. 
From Lemma~\ref{lem:x_supp}, $|x_j|^2, |x'_{j'}|^2 \geq \sfrac{1}{2}$. 
Since the $j^\tth$ and $j'\,^\tth$ components of $\p{x}-\p{x'}$ are $x_j$ and $-x'_{j'}$, we conclude $\|\p{x} - \p{x'}\|^2 \geq |x_j|^2 + |-x'_{j'}|^2 \geq \sfrac{1}{2} + \sfrac{1}{2} = 1$.

If $A=A'$, then $(\p{z},\p{t}) \neq (\p{z'},\p{t'})$ and $\|\p{x} - \p{x'}\|^2=\|\p{z} + \p{t} - (\p{z'} + \p{t'})\|^2$. 
From Lemma~\ref{lem:z_t_distance}, we conclude $\|\p{x} - \p{x'}\|^2=\|\p{z} + \p{t} - (\p{z'} + \p{t'})\|^2 \geq 1$. 


\emph{Upper bound:} Consider any $\p{z} \in \Mm^\na$ and $A \in \Ac$. Let $\p{t}=(0,\dots,0)^\tr$ and $\p{t'}=(\alpha,\alpha,0,\dots,0)^\tr$.
Consider the codewords $\p{x} = \psi_A(\p{z} + \p{t})$ and $\p{x'}=\psi_A(\p{z} + \p{t'})$.
Then $\dmin^2 \leq \|\p{x} - \p{x'}\|^2=\|\p{z} + \p{t} - (\p{z} + \p{t'})\|^2=\|\p{t} - \p{t'}\|^2=1$. 
\end{proof}


{
\subsubsection*{Nominal Coding Gain of $\Cn$}

Using~\eqref{eq:Pn} and Theorem~\ref{thm:mindist_Cn},
\begin{equation} \label{eq:delta_Cn_general} 
\delta(\Cn) = \frac{\dmin^2(\Cn)}{\Pn} = \frac{2}{\na (\, \energy(\Mm) + \energy(\Mm + \alpha)\,)}.
\end{equation}
Observe that $\delta(\Cn)$ is inversely proportional to the energies of the alphabets $\Mm$ and $\Mm+ \alpha$. A larger value of $M$ requires the use of a larger alphabet $\Mm$, which increases the energies $\mathcal{E}(\Mm)$ and $\mathcal{E}(\Mm + \alpha)$, and thereby reduces $\delta(\Cn)$.
From~\eqref{eq:Cn}, note that $|\Cn|=M^{\na}L 2^{\na-1}$ and the spectral efficiency is $\na \log_2 M + \log_2 L + \na - 1$.
Thus, for given values of $\nt,\na$ and spectral efficiency, we can maximize the nominal coding gain by using as small a value of $M$, or equivalently, as large a value of $L$, as possible.
Finally, note that when $\Mm$ is the modified square $M$-QAM,
using~\eqref{eq:energy_Mm} and~\eqref{eq:energy_Mm_alpha} in~\eqref{eq:delta_Cn_general}, 
\begin{equation} \label{eq:delta_Cn_squareQAM}
\textstyle \delta(\Cn) = \left[ {\na \left( \frac{M-1}{6} + \frac{1}{2} + \frac{1}{4\sqrt{M}} \right)} \right] ^{-1}. 
\end{equation}  
}

{
\subsubsection*{Decoding Complexity} 

The maximum-likelihood decoder computes $\arg \min_{\p{x} \in \Cn} \|\p{y} - \p{Hx}\|^2$, where $\p{y}$ and $\p{H}$ are channel output and channel matrix, respectively. 
From~\eqref{eq:psi_A} and Lemma~\ref{lem:x_supp}, \mbox{$\|\p{x}\|_0=\na$} for all \mbox{$\p{x} \in \Cn$}.
Section~V-C of~\cite{CSSS_TWC_16} shows that if every transmit vector $\p{x}$ satisfies \mbox{$\|\p{x}\|_0=\na$}, as is the case with our constellation $\Cn$, then the number of complex floating point operations (flops) required to compute \mbox{$\|\p{y} - \p{Hx}\|^2$} for all choices of $\p{x}$ is $2^\eta (2\nr (\na + 1) - 1)$, where $\eta = \log_2 |\Cn|$. 
Our receiver complexity is identical to the ML decoding complexity of GSM and the constellations from~\cite{CSSS_TWC_16,FRS_WCL_18} for the same of values of $\na,\nr$ and spectral efficiency $\eta$. 
}

\subsection{Comparison with GSM}

We compare the new constellation with GSM 
when the spectral efficiencies are equal.
For fairness, we will also assume that both schemes use the same number of antenna activation patterns $L$.
We will assume that $\Cn$ uses a modified square QAM alphabet of size \mbox{$M=2^{2m}$}, where \mbox{$m \geq 2$} is a positive integer. 
Then the spectral efficiency of $\Cn$ is $2m\na + \log_2 L + \na - 1$.
To achieve the same spectral efficiency, GSM must use two different complex alphabets --- it encodes one symbol using $2^{2m}$-ary square QAM, and the remaining $\na-1$ symbols using a $2^{2m+1}$-ary QAM. We assume that both the alphabets are subsets of $\Zb[i] + \alpha$, and the latter alphabet is a \emph{cross QAM}~\cite{VAK_TWC_05} of size $2^{2m+1}$  
(note that the energy of cross QAM is less than that of the rectangular QAM of size $2^{2m+1}$).
It is straightforward to show that the minimum distance of this GSM constellation is equal to $1$. 
The energy of the $2^{2m}$-ary square QAM is $(2^{2m}-1)/6$, and that of the $2^{2m+1}$-ary cross QAM is 
$(31 \times 2^{2m+1}/32 \,-\, 1)/6$, see~\cite{VAK_TWC_05}.
Using this we obtain the power of the GSM constellation $\Cs_\gsm$
\begin{equation*}
P(\Cs_\gsm) = \frac{2^{2m}-1}{6} + (\na-1)\frac{31 \times 2^{2m+1}/32 \,-\, 1}{6}, 
\end{equation*} 
and the nominal coding gain $\delta(\Cs_\gsm)=1/P(\Cs_\gsm)$.
Using this with~\eqref{eq:delta_Cn_squareQAM}, and after straightforward manipulations, 
\begin{equation} \label{eq:delta_Cn_Cgsm_ratio}
\frac{\delta(\Cn)}{\delta(\Cs_\gsm)}  = \frac{31}{16} ~\times~ \frac{1 - \frac{16}{31 M} - \frac{15}{31 \na}}{1 + \frac{2}{M} + \frac{3}{2M\sqrt{M}}}.
\end{equation} 
For large values of $M=2^{2m}$ and $\na$, 
the ratio $\delta(\Cn)/\delta(\Cs_\gsm)$ 
approaches \mbox{$31/16 \approx 1.94$}, which is a gain of $2.87$~dB.
Observe that the second term on the RHS of~\eqref{eq:delta_Cn_Cgsm_ratio} is an increasing function of $M$ and $\na$. 
Thus, substituting the smallest possible values for $M$ and $\na$, i.e., \mbox{$M=16$} and \mbox{$\na=2$}, in the RHS of~\eqref{eq:delta_Cn_Cgsm_ratio} shows that $\delta(\Cn)/\delta(\Cs_\gsm)$ is lower bounded by $60/49$ 
(approximately $0.86$~dB), for any $m \geq 2$ and $\na \geq 2$. 

\section{Simulation Results} \label{sec:simulation}

We now present simulation results to compare the transmit vector or codeword error rate (CER) of the new constellation with known schemes in the literature under maximum-likelihood decoding at the receiver. 
We do not consider schemes whose transmit vectors contain components with irrational real or imaginary parts. 
All simulations presented in this section use $\nt=4$. 

\begin{figure}[!t]
\centering
\includegraphics[width=3.2in]{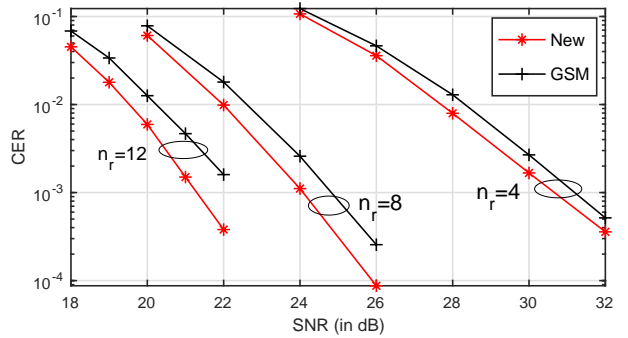}
\vspace{-2mm}
\caption{Comparison with GSM, \mbox{$\na=3$}, $\eta=22$~bits/sec/Hz.}
\label{fig:GSM}
\vspace{-2mm}
\end{figure} 

Fig.~\ref{fig:GSM} shows the comparison of the new scheme $\Cn$ with GSM for spectral efficiency $22$~bits/sec/Hz, $\na=3$ active antennas and for three different values of $\nr$, viz., $4$, $8$ and $12$. 
Both schemes use \mbox{$L=4$} antenna activation patterns. 
While the new scheme uses the modified square $64$-QAM alphabet and $2^{\na-1}=4$ translation patterns (as in Example~\ref{ex:Example2}), GSM modulates two of the three transmit symbols using cross $128$-QAM and the remaining symbol using square $64$-QAM. Fig.~2 shows that the new scheme is better, and the performance gap is larger for larger values of $\nr$. From~\eqref{eq:delta_Cn_Cgsm_ratio}, the nominal coding gain of $\Cn$ is larger than that of GSM by $1.92$~dB.
Note that~\cite{CSSS_TCOM_15,CSSS_TWC_16,FRS_WCL_18} do not provide constructions for $\na=3$.

\begin{figure}[!t]
\centering
\includegraphics[width=3.2in]{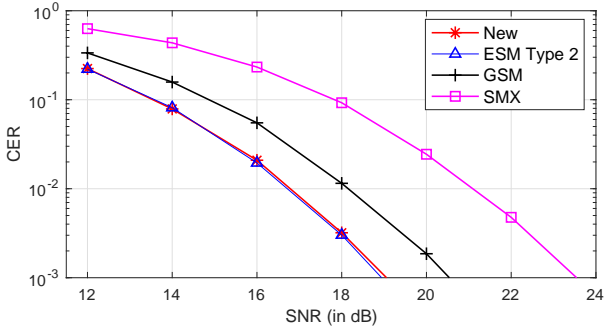}
\vspace{-2mm}
\caption{Comparison with ESM Type2, GSM and SMX.}
\label{fig:ESM2}
\vspace{-2mm}
\end{figure} 

Fig.~\ref{fig:ESM2} compares the new scheme with GSM, spatial multiplexing (SMX)~\cite{TsV_Cambridge}, and Enhanced Spatial Modulation (ESM)-Type2 from~\cite{CSSS_TWC_16} for \mbox{$\na=2$}, \mbox{$\nr=8$} and \mbox{$\eta=14$}~bits/sec/Hz.
Among the schemes presented in~\cite{CSSS_TWC_16} that do not perform coding across time, ESM-Type2 has the best error performance.
SMX uses two transmit antennas and cross $128$-QAM alphabet. GSM uses square $64$-QAM and \mbox{\emph{$L=4$}} antenna activation patterns.
ESM-Type2 exploits all possible, i.e., $\binom{\nt}{\na}=6$, pairwise antenna activations in addition to single-antenna activations.
For fairness, the new scheme $\Cn$ (see~\eqref{eq:Cn}) uses \mbox{$L=6$} activation patterns. We use a $37$-ary alphabet $\Mm$ satisfying~\eqref{eq:P1_P2}, along with $2^{\na-1}=2$ translation patterns, yielding $\eta$ slightly more than $14$~bits/sec/Hz.
Fig.~\ref{fig:ESM2} shows that ESM-Type2 and the new scheme have similar CER, and both of them perform better than GSM and SMX.
The nominal coding gains $\delta$ of ESM-Type2 and the new scheme are $0.0821$ and $0.079$, respectively, a loss of $0.17$~dB for the new scheme with respect to ESM-Type2.

\begin{figure}[!t]
\centering
\includegraphics[width=3.2in]{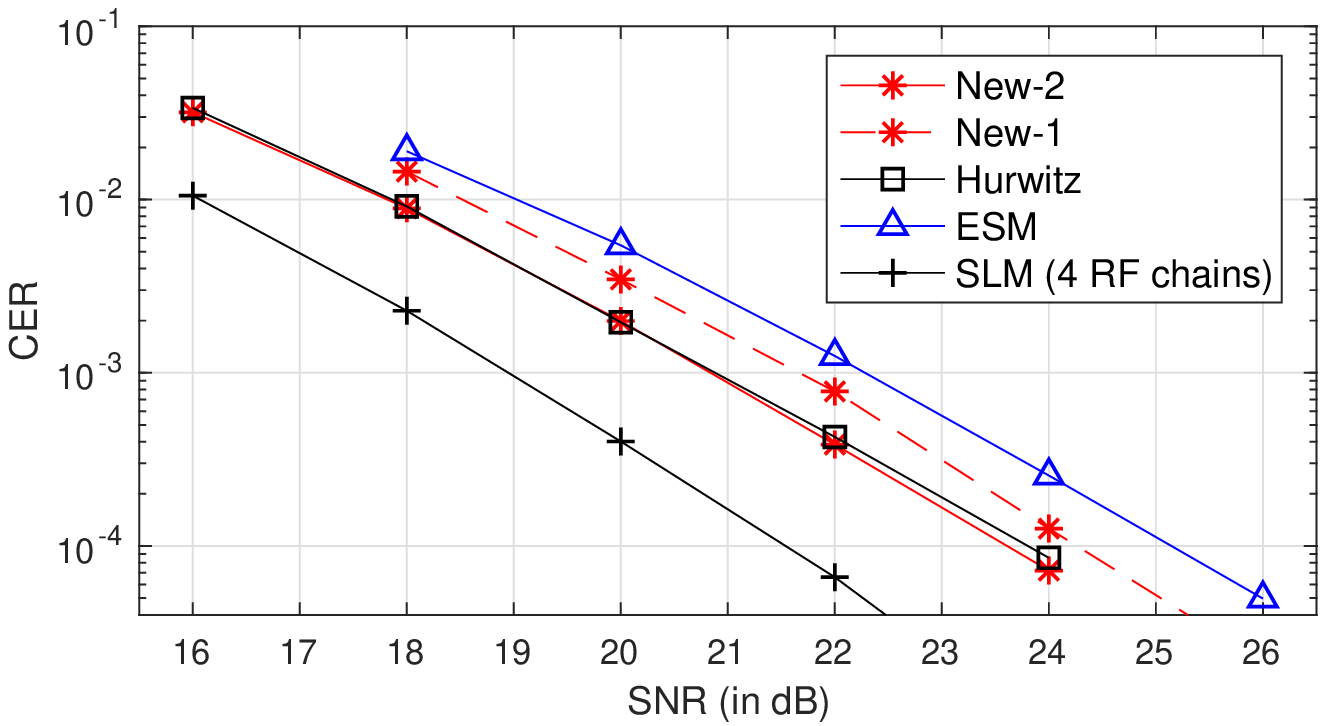}
\vspace{-2mm}
\caption{Comparison for \mbox{$\na=2$}, \mbox{$\nr=4$}, \mbox{$\eta \approx 11$}~bits/sec/Hz.}
\label{fig:Hurwitz}
\vspace{-3mm}
\end{figure} 

In Fig.~\ref{fig:Hurwitz} we consider $\na=2$, $\nr=4$ and spectral efficiencies close to $11$~bits/sec/Hz. We compare two new schemes, New-1 and New-2, with the Hurwitz integers based constellation from~\cite{FRS_WCL_18} and Enhanced Spatial Modulation (ESM) from~\cite{CSSS_TCOM_15}. The scheme from~\cite{FRS_WCL_18} uses $L=6$ antenna activation patterns, while ESM uses $6$ pairwise antenna activations and $4$ single-antenna activations, and both these schemes have $\eta=11$~bits/sec/Hz. 
The new scheme New-1 uses modified square $16$-QAM (see Fig.~\ref{fig:modified_16QAM}) and $L=4$ yielding $\eta=11$.
To be fair with respect to the number of antenna activation patterns, we include the performance of New-2, which uses $L=6$ and a $13$-point constellation $\Mm$ yielding $\eta=10.99$~bits/sec/Hz. 
Fig.~\ref{fig:Hurwitz} shows that New-1 performs better than ESM, even though it uses less number of antenna activations than ESM, and New-2 outperforms both ESM and the Hurwitz integer based constellation.
For completeness, Fig.~\ref{fig:Hurwitz} also shows the performance of spatial lattice modulation (SLM) based on Barnes-Wall lattice~\cite{CNL_TSP_18} with \mbox{$\eta=11$} that requires four transmit RF chains. 
All other schemes in Fig.~2 use \mbox{$\na=2$} RF chains.
The nominal coding gains $\delta$ of New-1, New-2, the Hurwitz integer scheme~\cite{FRS_WCL_18}, ESM~\cite{CSSS_TCOM_15} and SLM~\cite{CNL_TSP_18} are $0.1633$, $0.1985$, $0.1901$, $0.0773$ and $0.5314$, respectively. 

\section{Conclusion} \label{sec:conclusion}

We showed that carefully designed translation patterns, in addition to antenna activation patterns, can be used to encode information in GSM. 
Comparing with known schemes whose transmit vectors consist of only rational components, the new constellation performs nearly as well as ESM-Type2~\cite{CSSS_TWC_16}, and outperforms~\cite{CSSS_TCOM_15,FRS_WCL_18} and GSM in terms of error performance and nominal coding gain.
Explicit construction of ESM-Type2 is available only for limited choice of system parameters (only for $16$ and $64$-QAM alphabets, $\nt=4,8$ and $\na$ even).
In comparison, our framework provides constructions for any choice of $\nt$, $\na \geq 2$ and alphabet size $M$.


\bibliographystyle{IEEEtran}

\end{document}